\documentclass[12pt]{amsart}
\usepackage{amsthm}
\usepackage{amsmath}

\usepackage{breqn}
\usepackage{dsfont}
\usepackage{graphicx}
\author[K. Turek]{Krzysztof Turek}\address{Jagiellonian University\\ Institute of Mathematics\\ 30-348 Cracow, Poland}
\email{krzysztof.turek@uj.edu.pl}
\title{Option pricing in constant elasticity of variance model with liquidity costs.}
\predisplaypenalty
\postdisplaypenalty

\makeatother
\begin{document}
\baselineskip=17pt
\begin{titlepage}
\date{\today}
\begin{abstract}
Paper is based mainly on \cite{SemFin}. We generalize its thesis to constant elasticity model, which own previously used Black-Schoels model as a special case.\\The Goal of this article is to find optimal hedging strategy of European call/put option in illiquid environment. We understand illiquidity as a non linear transaction cost function depending only on rate of change of our portfolio. In case this function is quadratic, optimal policy is given by system of 3 PDE.\\
In addition we show, that for small $\epsilon$ costs of selling portfolio in time $T$ be important ($O(\epsilon)$) and shouldn't be neglected in Value function ($o(\epsilon^k)$- our result).

\end{abstract}
\maketitle

\theoremstyle{definition}
\newtheorem{definition}{definition}
\theoremstyle{plain}
\newtheorem{theorem}{theorem}
\theoremstyle{remark}
\newtheorem{collorary}{collorary}
\theoremstyle{remark}
\newtheorem{Lemma}{Lemma}
\end{titlepage}
\section{Introduction}
Problems with liquidity concerns people as long as the very basic concept of money, although most of models of option pricing still does not take into account this phenomena properly, without neglecting others. Great impact of liquidity on, non only option, pricing was shown in e.g. \cite{cenaPlynnosci} \cite{Longstaff}. During last decades there appeared numerous proposal, addressing this problem. Some of them assume direct, permanent effect of agent's actions on asset prices (e.g. \cite{wplywNaRynek}), but more include only temporal effects like in \cite{BSRozszerzony}\cite{SemFin}. Most authors from the second group tend to agree, that the the average price of each unit of asset is a convex function of quantity of bought or sold asset, with minimum at $0$. Typically this function is assumed to be exponential \cite{BSRozszerzony} or quadratic \cite{SemFin} with good results. The second function could be seen as a Taylor polynomial of degree two for the first function \cite{BSRozszerzony}. Whatever of this functions is used, even short burst of great volatility may generate huge liquidity cost for hedging investor. Thus in pricing options with illiquidity cost, modeling volatility is crucial.\\
In this papier we tried to merge the  model of illiquidity costs proposed in \cite{SemFin} and Constant Elasticy of Variance (CEV in short, Cox, 1975). CEV model is generalisation of Black-Schoels, improved by allowing non-zero correlation between observed price and volatility. In practice the price drops are connected with greater volatility in the near future, as a consequence of price and volatility are negatively correlated \cite{cont2001empirical}. Popular name of this phenomena is leverage effect. CEV model however exhibits purely in short horizon. Intuitively, there are other than price factors affecting volatility. One may go even further and try to use SABR model, however calculations are becoming hideous very fast and in some points different, more advanced technic is required, due to exogenous volatility process. The main difference to case considered here is the appearance of new variable $\sigma$ - volatility in HJB equation and accordingly changed infinitesimal generator $L$ (see e.g. \cite{SABRgeneratorL} for $L$ for SABR).\\
In CEV model asset price dynamic is given by stochastic differential equation:
\begin{equation}\label{CEV}
dS_t=\mu S_t dt+S_t^{1+\gamma}\sigma dW_{t}
\end{equation}
where $\mu$ and $\sigma$ are playing the same role as their counterparts in Black-Schoels model. Parameter $\gamma$ describes how volatility reacts on low and high asset price.
\subsection{Assumptions.}
We will focus on the case of $\gamma<0$, since many empirical evidences \cite{glosten1993relation}\cite{brandt2004relationship}\cite{cont2001empirical} have shown that the relationship between the stock price and its return
volatility is negative. However we must assume $\gamma\geq -\frac{1}{2}$. For simplicity we assume, that $\mu=0$. It can be done by constantly discounting data by constant, known from beginning $\mu\ne0$. New $S_t$ is old $S_t e^{-\mu t}$. Data processed in this way will fit to $\mu=0$ case. The last result may be easily obtained by simple applying Ito Lemma.\\
Basic assumptions are taken from \cite{SemFin}. Let us assume, that liquidity cost is proportional to asset price with constant  ($l(\cdot)$, $l:\mathds{R}\to\mathds{R}^{+}$) depending only on number of units sold or bought ($h_t$). We also assume, that price change connected with liquidity is temporal and market almost immediately goes back to its equilibrium i.e. $S_t$ before my deal. Therefore the cost of buying $h_t$ units in moment $t$ and market price $S_t$ is:
$$l(h_t)S_t+h_t S_t$$
Similarly total cash flow of realizing strategy $h_t, t\in[t_0,t_1]$ is given by:
$$\int_{t_0}^{t^1}l(h_t)S_t+h_t S_t$$
This idea mimics stock's exchange mechanism, where every bid is processed as set of bids on one share. Because every bid is matched with best remaining opposite offer, our following bid will be matched with worser opposite offer than previous one. In formal language l() is strictly convex and it has exactly one minimum - in $0$. For sake of simplicity we assume, that $l(0)=0$ and $l()$ is differentiable. This implies, that $l'()$ is increasing (convexity). The cost of illiquidity during realization of strategy $h_t, t\in[0,T]$ with stock prizes $S_t, t\in[0,T]$ is given by integral:

\begin{equation}\label{KosztP}
\int_{0}^{T}l(h_{t})S_t dt
\end{equation}

The number of shares in portfolio in time $t$ we denote by $H_t$. Previously introduced denotes remain in following relationship with this one:
\begin{equation}\nonumber
H_t=H_0+\int_0^t h_s ds
\end{equation}
Too prove any meaningful results we must make some assumptions on $H$. We require:
\begin{equation}\label{Hwar}
\begin{split}
&E(\int_0^TH_t^2 S_t^2dS_t)<+\infty\\
&E(\int_0^Tl(h_t) S_tdt)<+\infty
\end{split}
\end{equation}
The first assures value of portfolio $H$ will not become too big, which is also expected cost of buying this portfolio without cost of illiqudity. The second says not to change portfolio too quickly (infinite illiquidity cost). Both conditions are quite natural (they are required only to prove, that solution of final HJB PDE will be true value function). Hedging strategies in different classical models fulfill similar conditions. The set of all dynamics $h$ of $H$, which fulfill our assumptions, will be called $\mathcal{H}$.
Let $x_0$ be initial cash value. Classic (without cost of illiquity) value of entire portfolio in time $t$ is denoted as $\xi_t$. It does not include cost of illiquidity to avoid additional complications. $\xi_t$ is given by formula:
\begin{equation}\nonumber
\xi_t=H_0S_0+x_0+\int_0^t H_t dS_t
\end{equation}
We want to hedge a option, which pay $G(S_T)$ at $T$. Its price $q(S,t)$ at $t$ is assumed to be a expected value of payment at the end. It is also value of self-financing arbitrage strategy (without illiquidity cost). This relationship might be expressed as partial differential equation (infinitesimal generator in CEV):
 \begin{equation}\nonumber
L(q)=0
\end{equation}
\begin{equation}\label{warBrzeg}
q(\cdot,T)=G(\cdot)
\end{equation}
where:
\begin{equation}\label{OpL}
L(\cdot)=\frac{1}{2}\sigma^2 S^{2+2\gamma}\frac{\partial^2}{\partial S^2}+\frac{\partial}{\partial t}
\end{equation}
We denote $\theta_t$ as self-financing strategy ($\theta_t=\frac{\partial q}{\partial S}$) in perfectly liquid environment. This will be the base for our strategy in illiquid environment. In order to prove later results (that HJB solution is a value function), we must assume, that $\theta_t$ is bounded as follows:
\begin{equation}\label{BoundedTheta}
\exists_{C>0, \alpha>0}\forall_{S>0, t\in[0,T]}:|\theta(t,S)|<C(1+S^\alpha)
\end{equation}
Last result - fact, that $V$ is Landau small o of $\epsilon^k$ requires additional assumption:
\begin{equation}\label{BoundedThetaS}
\exists_{C_{1}>0, \beta>0}\forall_{S>0, t\in[0,T]}:|\frac{\partial\theta(t,S)}{S}|<C(1+S^\beta)
\end{equation}

This assumption holds for European call/put options in CEV model (see formula for prizing European options in \cite{cox1976valuation}).\\
\section{Value function and HJB equation.}
Clearly $\xi_t$ perfectly fitting to option's price is cautiously, rapidly changing strategy. This results in huge cost of illiquidity. To minimize the cost of iliquidity, we should pick constant $h_t\equiv 0$. However this strategy is pointless from point of view of hedging. To achieve satisfying results in terms of hedging and cost of liquidity the author of \cite{SemFin} propose minimization of:
\begin{equation}\label{funkcjaCel1}
\psi_{0}=\frac{1}{2}E((G(S_T)-\xi_{t})^2)+\int_{0}^Tl(h_t)S_t dt
\end{equation}
We will try to minimize this expression as well. We must point out, that our last result says basically, that for very small illiquidity cost (small $\epsilon$) cost of selling portfolio in time $T$ shouldn't be neglected. In this form it is hard to handle, thus it need a transformation into more appropriate form. We use the Ito formula and the Ito isometry:
\begin{equation}\nonumber
\begin{split}
&\psi_{0}=\frac{1}{2}(x_0+H_0S_0-q_0(S_0))^2 +\frac{1}{2}E(\int_0^T (\theta_t(S_t)-H_t)^2\sigma^2 S^{2+2\gamma}dt +\int_0^TS_tl(h_t)dt)\\ &=\frac{1}{2}(x_0+H_0S_0-q_0(S_0))^2 +\psi
\end{split}
\end{equation}
We could omit a factors independent of chosen strategy. It does not change optimal strategies, so we will minimize just $\psi$. We must point out, that constants like $H_0$ still appear in $\psi$ in hidden form, because $H_t=H_0+\int_0^t h_s ds$. Finally our cost function $V(H,S,t)$ for HJB equation is given by:
\begin{equation}\label{Vrow}
\inf_{h\in\mathcal{H}}\frac{1}{2}E(\int_t^T (\theta(u,(S_u))-H_u)^2\sigma^2 S^{2+2\gamma}_udu +\int_t^TS_ul(h_u)du|H_t=H,S_t=S)
\end{equation}
We remark absence of $S_T$ - dependant factor in \ref{Vrow}, thus final condition in Hamilton - Bellman - Jacobi equation is $V(H,S,T)=0$. Equation itself:
\begin{equation}\nonumber
0=\inf_h V_t+hV_H+\frac{1}{2}\sigma^2S^{2+2\gamma}V_{SS}+ \frac{1}{2}\sigma^2S^{2+2\gamma}(\theta(t,S)-H)^2+Sl(h)
\end{equation}
This equation is quite long, to simplify notation we will use operator $L$, defined in $\ref{OpL}$.
\begin{equation}\label{HJB}
0=\inf_h hV_H+L(V)+ \frac{1}{2}\sigma^2S^{2+2\gamma}(\theta(t,S)-H)^2+Sl(h)
\end{equation}
$h$  still is a function of time, as $V$ is function of time, initial share price, initial numer of shares. For the sake of simplicity we omit all this variables, everywhere we can. We noticed, that stochasticity has disappeared and we deal with deterministic partial differential equation.
It is possible to go further with general $l(\cdot)$, but it causes the appearance of a convex conjugate in partial differential equation, which can not be handled easily. For this reason we pick simple strictly convex $l(\cdot)$ with simple convex conjugate. Now we will try to calculate it as function of values $S, H, t$ and values of function $V$. This will allow us to substitute $h$ in equation above and, as a result, we will get some PDE in standard form to solve.

\subsection{Calculating $h$}
Let's take a look at a previous equation, it poses only two terms dependant of $h$: $Sl(h)$ we $hV_H$. Without any assumptions on $h$ it can be calculated as infimum in previous equation over every $S, H, t, V$.
\begin{equation}\nonumber
h=arg\min\limits_{\tilde{h}} Sl(h)+hV_H
\end{equation}
We can rewrite \ref{HJB} equation by using the convex conjugate of $-f$:
\begin{equation}\nonumber
0=L(V)+ \frac{1}{2}\sigma^2S^{2+2\gamma}(\theta-H)^2+S(-l)^{*}(V_H)
\end{equation}
Equation is hard to solve in general, so we will take particular $l$ to simplify calculations.
We are considering $V_H$ as a known for a moment and we assume differentiability of $l$. To calculate above infimum, we will use necessary condition for extremes for differentiable function. Since $l$ is also strictly convex and $hV_H$ is convex their sum is strictly convex and this condition is also sufficient.
\begin{equation}\nonumber
Sl'(h)+V_H=0
\end{equation}
\begin{equation}\label{hInv}
l'(h)=-\frac{V_H}{S}
\end{equation}
Derivative of strictly convex function must be increasing and have Darboux property, thus must be continuous and injective. If inverse function of derivative exist, it is also increasing and continuous. To assure existence of inverse function to derivative, we assume surjectivity of derivative. Logical choice (for $l'$) of easily invertible, continuous function is affine function. $l$ has minimum at $0$, so $l'(0)=0$. Only affine functions, which fulfil this condition are linear ones. Also $l(0)=0$, so $l$ must be of form:
\begin{equation}\nonumber
l(h)=\frac{\epsilon}{2}h^2
\end{equation}
It is easy to proof, that this function satisfies all assumption imposed on $l$.

\section{Derivation and solution of final PDE.}

By using \ref{hInv} we get the following:
\begin{equation}\label{ROWh}
h=-\frac{V_H}{S\epsilon}
\end{equation}
After substituting this to equation  \ref{HJB} we obtain:
\begin{equation}\label{PDE}
0= L(V)+ \frac{1}{2}(\sigma^2S^{2+2\gamma}(\theta-H)^2-\frac{V_H^2}{2S\epsilon})
\end{equation}

The appearance of quadratic forms of $V_H$ only and linear by another derivatives suggests, that $V$ might be dependent only quadratically on $H$. This supposition turn out to be true. Suppose:
\begin{equation}\label{VabcH}
V(H,S,t)=a(S,t)H^2+b(S,t)H+c(t,S)
\end{equation}
for some $a, b, c$ $C^2$ functions according to $S$ variable and $C^1$ in $t$. Now, equation \ref{PDE} assumes form:
\begin{equation}\nonumber
0= L(aH^2+bH+c)+ \frac{1}{2}(\sigma^2S^{2+2\gamma}(\theta-H)^2-\frac{(2aH+b)^2}{2S\epsilon})
\end{equation}
\begin{equation}\nonumber
0= L(a)H^2+L(b)H+L(c)+ \frac{1}{2}(\sigma^2S^{2+2\gamma}(\theta^2+2\theta H+H^2)-\frac{4a^2H^2+4abH+b^2}{2S\epsilon})
\end{equation}
Let us recall terminal condition from HJB $\forall_{H,S}:V(H,S,T)=0$. It implies the following conditions:
$$\forall_{S}: a(S,T)=0, b(S,T)=0, c(S,T)=0$$

By comprising factors at different powers of $H$ we get system of equation: 
\begin{equation}\label{Ra}
0=L(a)+\frac{1}{2}\sigma^2S^{2+2\gamma}-\frac{2a^2}{S\epsilon}
\end{equation}
\begin{equation}\label{Rb}
0=L(b)-\sigma^2S^{2+2\gamma}\theta-\frac{2ab}{S\epsilon}
\end{equation}
\begin{equation}\label{Rc}
0=L(c)+\frac{1}{2}\sigma^2S^{2+2\gamma}\theta^2-\frac{b^2}{2S\epsilon}
\end{equation}
If we have a solution $a$ of equation \ref{Ra} equation \ref{Rb} will be easy by using technique analogous to Feynman-Katz representation. Similarly we would process with \ref{Rc} by using $b$ in the same manner.

Term $a^2$ is rather problematic to deal with. For this reasons, one $a$ we consider as a  known and the other one as a requested. This reasoning lead us to the auxiliary PDE:
\begin{equation}\label{auxRa}
0=L(\widetilde{a})+\frac{1}{2}\sigma^2S^{2+2\gamma}-\frac{2\widetilde{a}\widehat{a}}{S\epsilon}
\end{equation}
If $\widetilde{a}$ is solution of this PDE and is equal to $\widehat{a}$ then $\widetilde{a}$ is also a solution of \ref{Ra}. Let us assume that $\widehat{a}\geq 0$. 

\subsection{Integrability of analogue of Feynman-Kac representation for auxiliary PDE for $a$}
To solve the auxiliary PDE \ref{auxRa} we will use reasoning similar to the Feynman-Katz representation. We can not use it directly since $S^{1+\gamma}$ is not Lipschitz function. We will calculate solution as in Feynman-Kac formula, but we will prove its properties directly. This results in a equation:
\begin{equation}\nonumber
\widetilde{a}(t,S)=\frac{1}{2}E(\int_t^T e^{-\int_t^u 2\widehat{a}(\nu,S_{\nu})\frac{d\nu}{\epsilon S_\nu}}\sigma^2S_u^{2+2\gamma}du)
\end{equation}
All terms are non-negative, continuous, thus $\widetilde{a}$ is also non-negative (possibly infinite). We must show, that right side of the equation is finite. Since exponent is always non-positive, whole exponential term is smaller than $1$.
\begin{equation}\nonumber
\frac{1}{2}E(\int_t^T e^{-\int_t^u 2\widehat{a}(\nu,S_{\nu})\frac{d\nu}{\epsilon S_\nu}}\sigma^2S_u^{2+2\gamma}du)\leq \frac{1}{2}E(\int_t^T \sigma^2S_u^{2+2\gamma}du)
\end{equation}
We can use Tonelli's theorem to interchange expected value and Lebesgue integral signs, if resulting equation is finite.
\begin{equation}\nonumber\frac{1}{2}E(\int_t^T \sigma^2S_u^{2+2\gamma}du)\stackrel{?}{=} \frac{1}{2}\int_t^T \sigma^2E(S_u^{2+2\gamma})du
\end{equation}
$S_u$ is CEV process, so its density is given by equation below \cite{schroder1989computing}. We will use original notation (only $\gamma$ will remain unchanged), because expressions are quite long.
$$k=\frac{1}{\sigma^2 4\gamma^2(u-t)}$$ 
$$x=kS_t^{-2\gamma}$$  
$$z=kS_u^{-2\gamma}$$
Also, we change variable ($dS_u=(-2\gamma)^{-1}k^{\frac{1}{-2\gamma}}w^{\frac{1+2\gamma}{-2\gamma}}dz$) in density:

\begin{equation}\nonumber
f(S_u|S_t,u>t)=\int_0^\infty 2\gamma k^{\frac{1}{-2\gamma}} (xz^{1-2\gamma})^{\frac{1}{-4\gamma}}\exp(-x-z)I_{\frac{1}{-2\gamma}}(2\sqrt{xz})(-2\gamma)^{-1}k^{\frac{1}{-2\gamma}}w^{\frac{1+2\gamma}{-2\gamma}}dz
\end{equation}

Now we can transform raw moment $E(S_u^{2+2\gamma})$ into form of integral:
\begin{equation}\nonumber
E(S_u^{2+2\gamma})=\int_0^\infty  (\frac{z}{k})^{\frac{2+2\gamma}{-2\gamma}} 2\gamma k^{\frac{1}{-2\gamma}} (xz^{1-2\gamma})^{\frac{1}{-4\gamma}} \exp(-x-z) I_{\frac{1}{-2\gamma}}(2\sqrt{xz}) (-2\gamma)^{-1}k^{\frac{1}{-2\gamma}} w^{\frac{1+2\gamma}{-2\gamma}}dz
\end{equation}
\begin{equation}\nonumber
E(S_u^{2+2\gamma})=\int_0^\infty  (\frac{z}{k})^{\frac{2+2\gamma}{-2\gamma}} 2\gamma (\frac{x}{z})^{\frac{1}{-4\gamma}} \exp(-x-z) I_{\frac{1}{-2\gamma}}2\sqrt{xz} dz
\end{equation}
\begin{equation}\nonumber
E(S_u^{2+2\gamma})=(\frac{x}{k})^{\frac{1}{-2\gamma}}\int_0^\infty  (\frac{z}{k})^{\frac{1+2\gamma}{-2\gamma}} 2\gamma (\frac{z}{x})^{\frac{1}{-4\gamma}} \exp(-x-z) I_{\frac{1}{-2\gamma}}(2\sqrt{xz}) dz
\end{equation}

This transformations are used in the detailed calculation of CEV European prising formula (more precisely $C_1$ term) in \cite{hsu2008constant}. This integral is $\frac{1+2\gamma}{-2\gamma}=\frac{1}{-2\gamma}-1$ moment of chi square distribution with $2-\frac{1}{\gamma}$ degrees of freedom and non centrality parametr $2x$. Moment generating function of non-centralized chi square distribution is:
$$
M(t;df,ncp)=\frac{\exp(\frac{ncp t}{1-2t})}{(1-2t)^{\frac{k}{2}}}$$
$n$ moment as a $n$-th derivative of $M$ at $t=0$. In this point $1-2t$ terms are equal $1$ and $\exp(\frac{ncp t}{1-2t})=1$. Thus $n$ moment is polynomial of order $n$ of variable $df$ and $ncp$. Coefficients are $\gamma$ deepened only. If $\frac{1}{-2\gamma}-1$ is not integer, we still can use some bigger integer and Holder inequality, $r>\frac{1}{-2\gamma}-1$:
$$E(X^r)\leq E(X^{\frac{1}{-2\gamma}-1})^{\frac{r}{\frac{1}{-2\gamma}-1}}$$
Every polynomial of degree $n$ can be majored by some constant $C$ (e.g. sum of absolute values of polynomial coefficients) multiplied by absolute value of highest order term $|x|^n$ plus $1$. In our case $x$ is non negative and absolute value can be dropped. We will slowly return to our old notations. For $\frac{1}{-2\gamma}-1$ integer we get:
\begin{equation}\nonumber
E(S_u^{2+2\gamma})\leq C(\frac{x}{k})^{\frac{1}{-2\gamma}}(x^{\frac{1}{-2\gamma}-1}+1)
\end{equation}
\begin{equation}\nonumber
E(S_u^{2+2\gamma})\leq C(S_t^{-2\gamma})^{\frac{1}{-2\gamma}}((kS_t^{-2\gamma})^{\frac{1}{-2\gamma}-1}+1)
\end{equation}
Let us remind, that $k$ includes time variable, which is important from our perspective. Terms like $\gamma$, $\sigma$ we add to new constant $C'$.
\begin{equation}\nonumber
E(S_u^{2+2\gamma})\leq \widehat{C}(S_t^{2+2\gamma}+1)((T-u)^{1+2\gamma}+1)
\end{equation}
Finally:
\begin{equation}\nonumber\frac{1}{2}\int_t^T \sigma^2E(S_u^{2+2\gamma})du\leq \frac{1}{2}\int_t^T \sigma^2 \widehat{C}(S_t^{2+2\gamma}+1)((T-u)^{1+2\gamma}+1)du
\end{equation}
Right side is finite (polynomial are easily integrable), thus assumptions of Tonelli theorem are fulfilled. Careful reader my noticed that using Holder inequality provide similar result. 
\begin{equation}\nonumber\frac{1}{2}E(\int_t^T \sigma^2S_u^{2+2\gamma}du)=\frac{1}{2} \sigma^2\widehat{C}(S_t^{2+2\gamma}+1)(\frac{(T-t)^{2+2\gamma}}{2+2\gamma}-\frac{(T-T)^{2+2\gamma}}{2+2\gamma}+T-t)
\end{equation}
By Tonelli theorem $\frac{1}{2}E(\int_t^T \sigma^2S_u^{2+2\gamma}du)$ is also finite and so is our candidate for auxiliary PDE solution: \begin{equation}\nonumber\frac{1}{2}E(\int_t^T e^{-\int_t^u 2\widehat{a}(\nu,S_{\nu})\frac{d\nu}{\epsilon S_\nu}}\sigma^2S_u^{2+2\gamma}du)\end{equation}
\subsection{Solving PDE for $a$, $b$ and $c$}
We check if following function (integrable) is indeed a solution of auxiliary PDE.
\begin{equation}\nonumber
\widetilde{a}(t,S)=\frac{1}{2}E(\int_t^T e^{-\int_t^u 2\widehat{a}(\nu,S_{\nu})\frac{d\nu}{\epsilon S_\nu}}\sigma^2S_u^{2+2\gamma}du)
\end{equation}
This function is $t$-differentiable, due Lemma 7.3.2 \cite{oksendal2003stochastic}. We denote $X_t=e^{\int_t^T 2\widehat{a}(\nu,S_{\nu})\frac{d\nu}{\epsilon S_\nu}}$, $Z_t=\frac{1}{2}\int_t^T e^{-\int_t^u 2\widehat{a}(\nu,S_{\nu})\frac{d\nu}{\epsilon S_\nu}}\sigma^2S_u^{2+2\gamma}du$. First of all let us calculate derivative of $X_tZ_t$ (it is well defined):
\begin{equation}\nonumber Z_t X_t=\frac{1}{2} \int_t^T e^{\int_u^T 2\widehat{a}(\nu,S_{\nu})\frac{d\nu}{\epsilon S_\nu}}\sigma^2S_u^{2+2\gamma}du\end{equation}
\begin{equation}\nonumber
d(Z_u X_u)=-\frac{1}{2} e^{\int_u^T 2\widehat{a}(\nu,S_{\nu})\frac{d\nu}{\epsilon S_\nu}}\sigma^2S_u^{2+2\gamma}du
\end{equation}

\begin{dmath*}
d(Z_u)=-(e^{-\int_u^T 2\widehat{a}(\nu,S_{\nu})\frac{d\nu}{\epsilon S_\nu}})\frac{1}{2} e^{\int_u^T 2\widehat{a}(\nu,S_{\nu})\frac{d\nu}{\epsilon S_\nu}}\sigma^2S_u^{2+2\gamma}du+\penalty10 2\widehat{a}(u,S_{u})\frac{du}{\epsilon S_u}e^{-\int_u^T 2\widehat{a}(\nu,S_{\nu})\frac{d\nu}{\epsilon S_\nu}}(\int_u^T \frac{1}{2} e^{\int_\xi^T 2\widehat{a}(\nu,S_{\nu})\frac{d\nu}{\epsilon S_\nu}}\sigma^2S_\xi^{2+2\gamma}d\xi)
\end{dmath*}
\begin{dmath*}
d(Z_u)=-\frac{1}{2}\sigma^2S_u^{2+2\gamma}du+\penalty10 2\widehat{a}(u,S_{u})\frac{du}{\epsilon S_u}(\int_u^T \frac{1}{2} e^{\int_u^\xi 2\widehat{a}(\nu,S_{\nu})\frac{d\nu}{\epsilon S_\nu}}\sigma^2S_\xi^{2+2\gamma}d\xi)
\end{dmath*}

We are using Dynkin's formula to $\widetilde{a}(t,S)$ to get $L(\widetilde{a}(t,S))$:
\begin{dmath*}
L(\widetilde{a}(t,S))= lim_{r\downarrow 0} \frac{E[\widetilde{a}(S_{t+r},t+r)|S_t=S]-E[\widetilde{a}(S_{t},t)|S_t=S]}{r} = lim_{r\downarrow 0}\frac{E[Z_{t+r}|S_t=S]-E[Z_t|S_t=S]}{r} \penalty10 = lim_{r\downarrow 0}\frac{ \int^{t+r}_t E(-\frac{1}{2}\sigma^2S_u^{2+2\gamma}+ 2\widehat{a}(u,S_{u})\frac{1}{\epsilon S_u}Z_t|S_t\penalty10000=S)}{r} \penalty10 = E(-\frac{1}{2}\sigma^2S_t^{2+2\gamma}+2\widehat{a}(t,S_{t})\frac{1}{\epsilon S_t}Z_t|S_t\penalty10000=S)\penalty10 =-\frac{1}{2}\sigma^2S_t^{2+2\gamma}+\penalty10 2\widehat{a}(t,S_{t})\frac{1}{\epsilon S_t} E(Z_t|S_t\penalty10000=S)
\end{dmath*}
Thus $\widetilde{a}=E(Z_t|S_t=S)$ is solution of auxiliary equation \ref{auxRa}.

\begin{equation}\nonumber
a(t,S)=\frac{1}{2}E(\int_t^T e^{-\int_t^u 2a(\nu,S_{\nu})\frac{d\nu}{\epsilon S_\nu}}\sigma^2S_u^{2+2\gamma}du)
\end{equation}
\begin{Lemma}
There exists unique solution of PDE \ref{Ra} with terminal condition:
\begin{equation}\nonumber
\alpha(T,\cdot)\equiv 0
\end{equation}
\end{Lemma}

\begin{proof}
Any fixed point of $\Psi$ function (definition below) is solution of functional equation given by  Feynman - Katz analogous. This means it must be solution of \ref{Ra}. Conversely, any solution of \ref{Ra} must be fixed point of $\Psi$.
\begin{equation}\label{PsiFunctional}
\Psi(\alpha(\cdot,\cdot))=\frac{1}{2}E(\int_t^T e^{-\int_t^u 2\alpha(\nu,S_{\nu})\frac{d\nu}{\epsilon S_\nu}}\sigma^2S_u^{2+2\gamma}du)
\end{equation}
where $\alpha$ suffices following terminal condition:
\begin{equation}\nonumber
\alpha(T,\cdot)\equiv 0
\end{equation}
To simplify notation we won't repeat $|S_t=S$ part below. It should be in any equation involving equation inside of $Psi$ and analogous of Feynman-Katz representation.
It is clear, that $\Psi(\cdot)\geq 0$. We construct sequence recursively:
$$a^{(0)}\equiv 0$$
$$a^{(n+1)}=\Psi(a^{(n)})$$
We noticed, that $\Psi$ is non-decreasing:
\begin{equation}\label{PsiIn}
x\geq y \Rightarrow \Psi(x) \leq \Psi(y) 
\end{equation}
Since arguments of $\Psi$ are functions, inequality on the left side must hold for all arguments of $x$, $y$, like for the right side. We must point out, that if the left side inequality is strict on some set of non zero Lebesgue measure, right side inequality is strict. By repetitive using this inequality we get:
\begin{equation}\nonumber
a^{(0)}=0\leq a^{(2)}\leq a^{(1)}
\end{equation}
and
\begin{equation}\nonumber
a^{(1)}\geq a^{(3)}\geq a^{(2)}
\end{equation}
$a^{(2n)}$ is increasing, $a^{(2n+1)}$ decreasing and:
\begin{equation}\nonumber
a^{(2n)}\leq a^{(2n+2)}\leq a^{(2n+1)}
\end{equation} 
\begin{equation}\nonumber
a^{(2n)}\leq a^{(2n+1)}\leq a^{(2n-1)}
\end{equation} 
The last two inequalities may be proven by mathematical induction (all in one go). Induction step is just using \ref{PsiIn} twice.\\
$a^{(2n+1)}$ is decreasing and bounded from below by $0$. Therefore this sequence must have a limit $\bar{a}$. Similarly  $a^{(2n)}$ is increasing and bounded by any term of $a^{(2n+1)}$, so also by $\bar{a}$. Because of this, $a^{(2n)}$ have a limit $\underbar{a}$ and $\underbar{a}\leq\bar{a}$. Due to Monotone Convergence Theorem we have $\Psi(\bar{a})=\underbar{a}$ and $\bar{a}=\Psi(\underbar{a})$.

\begin{equation}\nonumber
0=L(\bar{a})+\frac{1}{2}\sigma^2S^{2+2\gamma}+\frac{4\bar{a}\underbar{a}}{S\epsilon}
\end{equation}
\begin{equation}\nonumber
0=L(\underbar{a})+\frac{1}{2}\sigma^2S^{2+2\gamma}+\frac{4\bar{a}\underbar{a}}{S\epsilon}
\end{equation}
Subtracting by sides:
\begin{equation}\nonumber
0=L(\bar{a}-\underbar{a})\Leftrightarrow \forall_t E(\bar{a}(S_t,t)-\underbar{a}(S_t,t))\equiv const
\end{equation}
We recall terminal condition $a(S,T)=0$ and inequality $\bar{a}-\underbar{a}\geq 0$. By using this two and the above equation, we obtain $\bar{a}=\underbar{a}$. Therefore we have at least one solution.\\
Intuitively there should be only one solution for hedging problem. It comes to be true, what we will prove briefly. We suppose, there is some $\tilde{a}$ solution of PDE and fixed point of $\Psi$. Because $\tilde{a}=\Psi(\tilde{a})$, also $\tilde{a}\geq 0$. We recall, that $\Psi$ is non-decreasing, therefore:
\begin{equation}\nonumber
a^{(0)}=0\geq \tilde{a}=\Psi(\tilde{a})\geq \Psi(0)=a^{(1)}
\end{equation}
Again, we use mathematical induction for $n$.
\begin{equation}\label{SzacA}
a^{(2n)}\geq \tilde{a}\geq a^{(2n+1)}
\end{equation}
For $a^{(2n)}$, $ a^{(2n+1)}$ let's take limit as $n$ goes to infinity. At this point we use again Monotone Convergence Theorem to achieve $\bar{a}=\tilde{a}$.\\
\end{proof}
We can use analogous of Feynman - Katz representation for $b$ and $c$ as we did for $a$ in auxiliary PDE, due to condition on $\theta$ \ref{BoundedTheta}. There are no essential difference in profs, so we will skip them. Results are as follows:
\begin{equation}\nonumber
b(t,S)=E(\int_t^T e^{-\int_t^u 2a(\nu,S_{\nu})\frac{d\nu}{\epsilon S_\nu}}\theta(u,S_u)\sigma^2S_u^{2+2\gamma}du)
\end{equation}
\begin{equation}\nonumber
c(t,S)=\frac{1}{2}E(\int_t^T e^{-0}(\theta^2(u,S_u)\sigma^2S_u^{2+2\gamma}+ \frac{2b^2(u,S_{u})}{\epsilon S_u})du)
\end{equation}
\subsection{Solution of HJB PDE is true Value Function.}
Exponential functions, their sums, compositions and integrals are smooth ($C^\infty$), thus $a, b, c$ are smooth. Therefore $V=H^2a+Hb+c$ is also smooth (on $H$ variable too). We have unique solutions $a$, $b$, $c$, which also fulfil same kind of condition \ref{BoundedTheta}. put on $\theta$. It can be shown directly from above equations. Entire $exp (..)$ is bounded by $1$ from the above and non negative. $\theta$ and $\sigma^2S_u^{2+2\gamma}$ also fulfills this condition and integrals over finite interval (due to Holder inequality) preserve this condition. Moreover, if we use $a^{(2n)}$ or $a^{(2n+1)}$ instead of $a$ we achieve upper and lower bonds for $b$ and $c$. It follows from \ref{SzacA}. and monotonic nature regarding $a$ (one could prove it in the same way we did for $\Psi$) of the equations for $b$ and $c$.\\
Let us recall equation for $V(H,S,t)$ \ref{Vrow}:
\begin{equation}\nonumber
\inf_{h\in\mathcal{H}}\frac{1}{2}E(\int_0^T (\theta(t,(S_t))-H_t)^2\sigma^2 S^{2+2\gamma}dt +\int_0^TS_tl(h_t)dt|H_0=H,S_0=S)
\end{equation}
Clearly, for any $h\in\mathcal{H}$ and $H$ connected to it following inequality holds:
\begin{equation}\nonumber
\begin{split}
&V(H_t,S_t,t)\leq \frac{1}{2}E(\int_t^T (\theta(u,(S_u))-H_u)^2\sigma^2 S^{2+2\gamma}_udu +\int_t^TS_ul(h_u)du|H_t=H,S_t=S|\mathcal{F}_t)\\
&\frac{1}{2}E(\int_0^T 2(\theta(u,(S_u))^2+H_u^2)\sigma^2 S^{2+2\gamma}_udu +\int_0^TS_ul(h_u)du|H_t=H,S_t=S|\mathcal{F}_t)
\end{split}
\end{equation}
Now, we use \ref{Hwar} (condition for being in $\mathcal{H}$):
\begin{equation}\nonumber
\begin{split}
&\int_0^TH_t^2 S_t^2dS_t<+\infty\\
&\int_0^Tl(h_t) S_tdt<+\infty
\end{split}
\end{equation}
and condition on $\theta$ \ref{BoundedTheta}:
\begin{equation}\nonumber
\exists_{C>0, \alpha>0}\forall_{S>0, t\in[0,T]}:|\theta(t,S)|<C(1+S^\alpha)
\end{equation}
If we sum up all these conditions and inequality above, we have process $V(H_t,S_t,t)$ bounded by uniformly integrable martingale. In fact conditions condition for being in $\mathcal{H}$ was needed only to ensure that these integrals exist and are finite. Strategies with infinite integrals won't achieve infimum, because constant strategy has finite integrals. Thus in fact we find infimum over larger set. From Ito lemma following process is local martingale plus non decreasing term (Doob–Meyer decomposition):
\begin{equation}\nonumber
Y_t=\frac{1}{2}E(\int_0^t (\theta(u,(S_u))-H_u)^2\sigma^2 S^{2+2\gamma}_udu +\int_0^tS_ul(h_u)du|H_t=H,S_t=S)+V(t,H_t,S_t)
\end{equation}
This process is connected with sequence of stoping times $\tau_m$ approaching $T$, which reduced stopped local martingale to martingale. Due to non decreasing term we have inequality between this stopped by $\tau_m$ $Y_t$ at $0$ and at $T$ ($\tau_m\land T=\tau_M$, because $\tau_m\leq T$ and $\tau_m\land 0=0$, because $\tau_m\geq 0$):
\begin{equation}\nonumber
V(0,H_0,S_0)\leq\frac{1}{2}E(\int_0^{\tau_m} (\theta(u,(S_u))-H_u)^2\sigma^2 S^{2+2\gamma}_udu +\int_0^{\tau_m}S_ul(h_u)du|H_t=H,S_t=S)+V(\tau_m,H_{\tau_m},S_{\tau_m})
\end{equation}
Let us note $V(T,H_{T},S_{T})=0$ from its definition. As $\tau_m$ approaches $T$ integrals monotonically approaches (in $L^1$):
\begin{equation}\nonumber
V(0,H_0,S_0)\leq\frac{1}{2}E(\int_0^T (\theta(u,(S_u))-H_u)^2\sigma^2 S^{2+2\gamma}_udu +\int_0^TS_ul(h_u)du|H_t=H,S_t=S)
\end{equation}

We show that optimal control $h$ always exists and $V$ solve HJB equation uniquely in set of appropriately smooth $C^n$ on each variable, then $V$ must be a true value function.

\section{Bounds of the function $V$}
Let us remind definition of value function $V$ (\ref{Vrow}):
\begin{equation}\nonumber
\inf_{h\in\mathcal{H}}\frac{1}{2}E(\int_t^T (\theta(u,(S_u))-H_u)^2\sigma^2 S^{2+2\gamma}_udu +\int_t^TS_ul(h_u)du|H_t=H,S_t=S)
\end{equation}
All terms are nonnegative, thus the function $V$ is bounded from below by $0$. Infimum over all $h$ must be smaller than value for a any particular $h$, thus value of above equation with a chosen $h$ is upper bound of $V$. In this moment we could write down this $h$ and calculate this bound, but we want to show a some intuitive reason to choose this one.\\
\subsection{Intuition behind choice of $\bar{h}$} 
HJB equation for our problem is following \ref{HJB}:
\begin{equation}\nonumber
0=\inf_h hV_H+L(V)+ \frac{1}{2}\sigma^2S^{2+2\gamma}(\theta(t,S)-H)^2+Sl(h)
\end{equation}
The term $L(V)$ describe change of expected value of the function $V$ in time. In ideal case this term is a $0$ or at least very small. We will assume $L(V)=0$ this for a moment. Additionally we substitute $l(\bar{h})=\frac{\epsilon \bar{h}^2}{2}$ into the equation.
\begin{equation}\nonumber
0=\inf_{\tilde{h}} \tilde{h}V_H+ \frac{1}{2}\sigma^2S^{2+2\gamma}(\theta(t,S)-H)^2+\frac{S\epsilon \tilde{h}^2}{2}
\end{equation}
We find the $\bar{h}$ in the same way as previously the $h$:
\begin{equation}\nonumber
\bar{h}=arg\min_{\tilde{h}} \tilde{h}V_H+L(V)+ \frac{1}{2}\sigma^2S^{2+2\gamma}(\theta(t,S)-H)^2+\frac{S\epsilon \tilde{h}^2}{2}
\end{equation}
Quadratic functions are convex, so they have one extremum and it is minimum. we calculate derivative of term inside of a argmin in respect to $\tilde{h}$.
\begin{equation}\nonumber
0=V_H+S\epsilon \bar{h}
\end{equation}
\begin{equation}\nonumber
\bar{h}=-\frac{V_H}{S\epsilon} 
\end{equation}
In this form $\bar{h}$ still have a term $V_H$. If we use this one we will obtain term $V_H^2$ in PDE later. To avoid problems we have encounter earlier, we will use HJB equation (in assumed in this section form, not true one) again. After substituting a $\bar{h}$, it looks as follows: 
\begin{equation}\nonumber
0=-\frac{V_H}{S\epsilon}V_H+ \frac{1}{2}\sigma^2S^{2+2\gamma}(\theta(t,S)-H)^2+\frac{S\epsilon (-\frac{V_H}{S\epsilon})^2}{2}
\end{equation}
\begin{equation}\nonumber
\frac{\bar{h}^2 S\epsilon}{2}=\frac{V_H^2}{2S\epsilon}= \frac{1}{2}\sigma^2S^{2+2\gamma}(\theta(t,S)-H)^2
\end{equation}
\begin{equation}\nonumber
\bar{h}^2=\frac{\sigma^2S^{2+2\gamma}(\theta(t,S)-H)^2}{S\epsilon}
\end{equation}
\begin{equation}\nonumber
\bar{h}=\sigma S^{\frac{1}{2}+\gamma}(\theta(t,S)-H)\epsilon^{-\frac{1}{2}}
\end{equation}
This $\bar{h}$ will be used to calculate upper bound.
\subsection{Upper bound}
We calculate value $V$ of strategy $\bar{h}$ by substituting to original equation \ref{funkcjaCel1}:
\begin{equation}\label{FK2D}
v(H,S,t|\bar{h})= E(\int_t^T \sigma^2S^{2+2\gamma}_\nu(\theta(\nu,S_\nu)-H_\nu)^2 d\nu|H_t=H,S_t=S)
\end{equation}Let us remind:
\begin{equation}\nonumber dH_\nu=\bar{h}_\nu d\nu=\sigma S^{\frac{1}{2}+\gamma}(\theta(\nu,S_\nu)-H_\nu)\epsilon^{-\frac{1}{2}}d\nu
\end{equation}

Main problem with determining limiting behavior of $V$ regardless to $\epsilon$ lies in $H_{\nu}$, as it contains $\epsilon$ in hidden form. In order to get rid of a problematic term we different terms below integral, build and solve SDE given in that way and finish calculations. Let us define $Y_t$ as follows:
\begin{equation}\nonumber
Y_{\nu}=\theta(\nu,S_\nu)-H_\nu
\end{equation}
$S_\nu, H_\nu$ still  have a initial conditions $H_t=H,S_t=S$. Before we start calculating $dY_t$ is important to calculate a derivative of $\theta(\nu,S_{\nu})=\frac{\partial q_\nu}{\partial S}$.
\begin{equation}\nonumber
L(\theta)=L(\frac{\partial q_\nu}{\partial S})=\frac{\partial L(q_\nu)}{\partial S}=\frac{\partial 0}{\partial S}=0
\end{equation}
This is implies ($L$ is infinitesimal generator of $S_{\nu}$), that $d\theta$ have form:
\begin{equation}\nonumber
d\theta(\nu,S_{\nu})=\sigma S^{1+\gamma}_\nu\frac{\partial \theta}{\partial S}(\nu,S_{\nu})dW_{\nu}
\end{equation}
Stochastic derivative of $Y_t$ is given by:
\begin{equation}\nonumber
dY_\nu=\sigma S^{1+\gamma}_\nu\frac{\partial \theta(\nu,S_\nu)}{\partial S}dW_\nu-\sigma S^{\frac{1}{2}+\gamma}_{\nu}(\theta(\nu,S_\nu)-H_\nu)\epsilon^{-\frac{1}{2}}d\nu
\end{equation}
Now we write this down using $Y_t$ notation:
\begin{equation}\nonumber
dY_\nu=\sigma S^{1+\gamma}_{\nu}\frac{\partial \theta(\nu,S_\nu)}{\partial S}dW_\nu-\sigma S^{\frac{1}{2}+\gamma}_{\nu}\epsilon^{-\frac{1}{2}}Y_{\nu}d\nu
\end{equation}
This stochastic (inhomogeneous) linear differential equation might be solved (uniquely) in a standard way (method of variation of parameters for more details see e.g. \i{5.1 Examples and Some Solution Methods} \cite{oksendal2003stochastic}):
\begin{dmath*}
Y_{\xi}=C(H,S,t)\exp(-\sigma\epsilon^{-\frac{1}{2}}\int_t^{\xi} S^{\frac{1}{2}+\gamma}_{\nu}d\nu) + \penalty10\sigma \int_t^{\xi}S^{1+\gamma}_{\nu}\frac{\partial \theta(\nu,S_\nu)}{\partial S}\exp(-\sigma\epsilon^{-\frac{1}{2}}\int_{\nu}^{\xi} S^{\frac{1}{2}+\gamma}_{\mu}d\mu)dW_\nu-\frac{1}{2}\sigma^2 \int_t^{\xi}S^{2+2\gamma}_{\nu}(\frac{\partial \theta(\nu,S_\nu)}{\partial S})^2\exp(-\sigma\epsilon^{-\frac{1}{2}}\int_{\nu}^{\xi} S^{\frac{1}{2}+\gamma}_{\mu}d\mu)d\nu
\end{dmath*}
where $C(H,S,t)$ is given by initial condition. In our problem:
\begin{equation}\nonumber\nonumber
C(H,S,t)=\sigma S^{1+\gamma}(\theta(t,S)-H)
\end{equation}
We rewrite equation \ref{FK2D} in $Y_t$ notation. This equation does not contain $H_{\nu}$, so will skip conditional expectation regarding $H_t=H$. To simplify notation we will skip $S_t=S$ part as well.
\begin{dmath}\label{vMale}
v(H,S,t|\bar{h})=\penalty10 E(\int_t^T \sigma^2S^{2+2\gamma}_\xi Y_{\nu}^2 d\nu|H_t=H,S_t=S)=\penalty10 E(\int_t^T \sigma^2S^{2+2\gamma}_\xi (C(H,S,t)\exp(-\sigma\epsilon^{-\frac{1}{2}}\int_t^{\xi} S^{\frac{1}{2}+\gamma}_{\nu}d\nu)+\penalty10 \sigma \int_t^{\xi}S^{1+\gamma}_{\nu}\frac{\partial \theta(\nu,S_\nu)}{\partial S}\exp(-\sigma\epsilon^{-\frac{1}{2}}\int_{\nu}^{\xi} S^{\frac{1}{2}+\gamma}_{\mu}d\mu)dW_\nu)^2 d\xi)
\end{dmath}
We will use simple inequality $(a+b)^2\leq (a+b)^2+(a-b)^2=2a^2+2b^2$:
\begin{dmath}\label{vBound}
v(H,S,t|\bar{h})\leq\penalty10 2E(\int_t^T \sigma^2S^{2+2\gamma}_\xi C(H,S,t)^2(\exp(-\sigma\epsilon^{-\frac{1}{2}}\int_t^{\xi} S^{\frac{1}{2}+\gamma}_{\nu}d\nu))^2 d\xi)+ \penalty10 2E(\int_t^T \sigma^2S^{2+2\gamma}_\xi(\sigma \int_t^{\xi}S^{1+\gamma}_{\nu}\frac{\partial \theta(\nu,S_\nu)}{\partial S}\exp(-\sigma\epsilon^{-\frac{1}{2}}\int_{\nu}^{\xi} S^{\frac{1}{2}+\gamma}_{\mu}d\mu)dW_\nu)^2 d\xi)
\end{dmath}
First term is integral regarding temporal variable and its integrand is non negative, thus we can neglect exponential term with negative exponent, since is smaller then $1$:
\begin{dmath}\label{firstTerm}
E(\int_t^T \sigma^2S^{2+2\gamma}_\xi C(H,S,t)^2(\exp(-\sigma\epsilon^{-\frac{1}{2}}\int_t^{\xi} S^{\frac{1}{2}+\gamma}_{\nu}d\nu))^2 d\xi)\leq  E(\int_t^T \sigma^2S^{2+2\gamma}_\xi C(H,S,t)^2 d\xi)=E((\int_t^T  \sigma S^{1+\gamma}_\xi C(H,S,t) dW_{\xi})^2)= E((\int_t^T C(H,S,t) dS_{\xi})^2)=\penalty10 E(C^2(H,S,t)(S_T-S_t)^2)<\infty
\end{dmath}
In above calculation we used Ito-isometry as well. Ito-isometry assumption are fulfilled since CEVs distribution is non central chi square \cite{cox1976valuation} (with appropriate parameters) and non central chi square have finite non central moments. Second term of right side of equation \ref{vBound} could transform by Ito-isometry again, but in different way:
\begin{dmath*}
E(\int_t^T \sigma^2S^{2+2\gamma}_\xi(\sigma \int_t^{\xi}S^{1+\gamma}_{\nu}\frac{\partial \theta(\nu,S_\nu)}{\partial S}\exp(-\sigma\epsilon^{-\frac{1}{2}}\int_{\nu}^{\xi} S^{\frac{1}{2}+\gamma}_{\mu}d\mu)dW_\nu)^2 d\xi)=\penalty10 E(\int_t^T (\int_t^{\xi} \sigma^2S^{1+1\gamma}_\xi S^{1+\gamma}_{\nu}\frac{\partial \theta(\nu,S_\nu)}{\partial S}\exp(-\sigma\epsilon^{-\frac{1}{2}}\int_{\nu}^{\xi} S^{\frac{1}{2}+\gamma}_{\mu}d\mu)dW_\nu)^2 d\xi)
\end{dmath*}
We are moving expectation under integral by $d\xi$, then use Ito-isometry to a inner integral by $d S_\nu$. Both of the theorems require a $L^2$ integrand. This can be achieved in the same way we used previously with first term in \ref{firstTerm} combined with assumed condition on $\frac{\partial \theta(\nu,S_\nu)}{\partial S}$ in \ref{BoundedThetaS}.
\begin{dmath*}
(..)=E(\int_t^T \int_t^{\xi} \sigma^4S^{2+2\gamma}_\xi S^{2+2\gamma}_{\nu}(\frac{\partial \theta(\nu,S_\nu)}{\partial S})^2\exp(-2\sigma\epsilon^{-\frac{1}{2}}\int_{\nu}^{\xi} S^{\frac{1}{2}+\gamma}_{\mu}d\mu)d\nu d\xi)
\end{dmath*}
We use same trick with exponential term as in \ref{firstTerm} and again move expectation sign under integral use Ito-isometry, but in other way. Then move expectation and use Ito-isometry once more.
\begin{dmath*}(..)\leq E(\int_t^T \sigma^2S^{2+2\gamma}_\xi\sigma \int_t^{\xi}S^{2+2\gamma}_{\nu}(\frac{\partial \theta(\nu,S_\nu)}{\partial S})^2d\nu d\xi)=\penalty10 E((\int_t^T \int_t^{\xi}\frac{\partial \theta(\nu,S_\nu)}{\partial S}dS_\nu dS_{\xi})^2)
\end{dmath*}
For similar usage of (iterated) Ito-isometry see \cite{oksendal1997malliavin}.
\begin{dmath*}
E((\int_t^T \int_t^{\xi}\frac{\partial \theta(\nu,S_\nu)}{\partial S}dS_\nu dS_{\xi})^2)=\penalty10 E((\int_t^T (\theta(\xi,S_\xi)-\theta(t,S_t)) dS_{\xi})^2)=\penalty10 E((\int_t^T (\theta(\xi,S_\xi)-\theta(t,S_t)) dS_{\xi})^2)=\penalty10 E(q(T,S_T)-q(t,S)-(S_T-S)\theta(t,S))<+\infty
\end{dmath*}
Since CEV distribution is non-central chi square \cite{cox1976valuation} (with appropriate parameters) and non central chi square have finite non central moments, thus right side of equation is finite. Formal prof is almost identical to we show to prove, that auxiliary PDE solution candidate is finite everywhere. 
\subsection{$V$ is $o(\epsilon^k)$}
Let us remind formula for $v(H,S,t|\bar{h})$ given in \ref{vBound}:
\begin{dmath*}
v(H,S,t|\bar{h})\leq\penalty10 2E(\int_t^T \sigma^2S^{2+2\gamma}_\xi C(H,S,t)^2(\exp(-\sigma\epsilon^{-\frac{1}{2}}\int_t^{\xi} S^{\frac{1}{2}+\gamma}_{\nu}d\nu))^2 d\xi)+ \penalty10 2E(\int_t^T \sigma^2S^{2+2\gamma}_\xi(\sigma \int_t^{\xi}S^{1+\gamma}_{\nu}\frac{\partial \theta(\nu,S_\nu)}{\partial S}\exp(-\sigma\epsilon^{-\frac{1}{2}}\int_{\nu}^{\xi} S^{\frac{1}{2}+\gamma}_{\mu}d\mu)dW_\nu)^2 d\xi)
\end{dmath*}
Both terms in equation are an expressions of type $\exp(-D\epsilon^{-\frac{1}{2}})$, where $D$ is a nonnegative, not $\epsilon$-dependent function. As a $\epsilon\downarrow 0$, $-D\epsilon^{-\frac{1}{2}}\downarrow-\infty$, thus $\exp(-D\epsilon^{-\frac{1}{2}})\downarrow 0$.  In last subsection we prove an integrals of expressions above are finite, thus due to Lebesgue's monotone convergence theorem $v(H,S,t|\bar{h})\downarrow 0$. Since $v(H,S,t|\bar{h})\geq V(H,S,t)\geq 0$, also $V(H,S,t)\downarrow 0$.\\ For any chosen $\forall_{k\in\mathds{N},k>0}$ we try to show that $v(H,S,t|\bar{h})\in o(\epsilon^{\frac{k}{2}})$. We could use L'Hôpital's rule:
\begin{dmath*}
\lim_{\epsilon\downarrow 0}\frac{\exp(-D\epsilon^{-\frac{1}{2}})}{\epsilon^{\frac{k}{2}}}=\lim_{\epsilon\downarrow 0}\frac{\epsilon^{-\frac{k}{2}}}{\exp(D\epsilon^{-\frac{1}{2}})}=\lim_{\epsilon\downarrow 0}
\frac{\frac{k}{2}\epsilon^{-\frac{k}{2}-1}}{-\frac{D}{2}\epsilon^{-1-\frac{1}{2}}\exp(D\epsilon^{-\frac{1}{2}})}=\lim_{\epsilon\downarrow 0}
\frac{\frac{k}{2}\epsilon^{-\frac{k+1}{2}}}{-\frac{D}{2}\exp(D\epsilon^{-\frac{1}{2}})}
\end{dmath*}
Thus if apply L'Hôpital's rule $k-1$ more time we achieve $\exp(-D\epsilon^{-\frac{1}{2}})\in o(\frac{k}{2})$. Using one more time Lebesgue's dominated convergence theorem we get appropriate result.
\bibliographystyle{plain}
\bibliography{PlynnoscOpcjeEng}

\begin{thebibliography}{10}

\bibitem{brandt2004relationship}
Michael~W Brandt and Qiang Kang.
\newblock On the relationship between the conditional mean and volatility of
  stock returns: A latent var approach.
\newblock {\em Journal of Financial Economics}, 72(2):217--257, 2004.

\bibitem{cont2001empirical}
Rama Cont.
\newblock Empirical properties of asset returns: stylized facts and statistical
  issues.
\newblock {\em Quantitative finance}, 2001.

\bibitem{cox1976valuation}
John~C Cox and Stephen~A Ross.
\newblock The valuation of options for alternative stochastic processes.
\newblock {\em Journal of financial economics}, 3(1):145--166, 1976.

\bibitem{glosten1993relation}
Lawrence~R Glosten, Ravi Jagannathan, and David~E Runkle.
\newblock On the relation between the expected value and the volatility of the
  nominal excess return on stocks.
\newblock {\em The journal of finance}, 48(5):1779--1801, 1993.

\bibitem{hsu2008constant}
Ying-Lin Hsu, TI~Lin, and CF~Lee.
\newblock Constant elasticity of variance (cev) option pricing model:
  Integration and detailed derivation.
\newblock {\em Mathematics and Computers in Simulation}, 79(1):60--71, 2008.

\bibitem{wplywNaRynek}
Hong Liu and Jiongmin Yong.
\newblock Option pricing with an illiquid underlying asset market.
\newblock {\em Journal of Economic Dynamics and Control}, 2005.

\bibitem{Longstaff}
Francis~A. Longstaff.
\newblock Option pricing and the martingale restrictionauthor.
\newblock {\em The Review of Financial Studies}, 1995.

\bibitem{cenaPlynnosci}
R.~Eldor M.~Brenner and S.~Hauser.
\newblock The price of options illiquidity.
\newblock {\em The Journal of Finance}, 2001.

\bibitem{oksendal1997malliavin}
Bernt {\O}ksendal.
\newblock Malliavin calculus.
\newblock 1997.

\bibitem{oksendal2003stochastic}
Bernt {\O}ksendal.
\newblock {\em Stochastic differential equations}.
\newblock Springer, 2003.

\bibitem{SemFin}
L.~C.~G. Rogers and Surbjeet Singh.
\newblock The cost of illiquidity and its effects on hedging.
\newblock {\em Mathematical finance}, 2010.

\bibitem{schroder1989computing}
Mark Schroder.
\newblock Computing the constant elasticity of variance option pricing formula.
\newblock {\em the Journal of Finance}, 44(1):211--219, 1989.

\bibitem{BSRozszerzony}
P.~Protter U.~Çetin, R.~Jarrow and M.~Warachka.
\newblock Pricing options in an extended black scholes economy with
  illiquidity: Theory and empirical evidence.
\newblock {\em The Review of Financial Studies}, 2006.

\bibitem{SABRgeneratorL}
Qi~Wu.
\newblock Series expansion of the sabr joint density.
\newblock {\em Mathematical Finance}, 22(2):310--345, 2012.

\end{thebibliography}
\end{document}